\documentclass{article}
\usepackage{amsthm}
\usepackage{geometry}

\def\MINE{1}

\usepackage{ifthen}
\usepackage{latexsym,amssymb,amsmath,mathdots}
\usepackage{graphics}
\usepackage{epsfig}
\usepackage{color}
\usepackage[dvipsnames]{xcolor}
\usepackage{graphicx}
\usepackage{tcolorbox}
\usepackage[colorlinks,citecolor=blue]{hyperref}
\usepackage{bm}
\usepackage{shuffle}   %
\usepackage{xspace}
\usepackage[OMLmathsfit]{isomath}

\usepackage{tikz}     
\usetikzlibrary{%
  arrows,%
  positioning,%
  decorations.pathmorphing,%
  decorations.pathreplacing,%
  automata,%
  shapes.geometric,%
}

\if\MINE1
\newtheorem{theorem}{Theorem}[section]

\newtheorem{definition}[theorem]{Definition}
\newtheorem{lemma}[theorem]{Lemma}
\newtheorem{corollary}[theorem]{Corollary}
\theoremstyle{definition}
\newtheorem{remark}[theorem]{Remark}

\newtheorem{example}[theorem]{Example}

\else
\fi

\newcommand\pnsi{\par\indent}
\newcommand\pnsn{\par\noindent}
\newcommand\pssi{\par\smallskip\indent}

\newcommand\pssn{\par\smallskip\noindent}
\newcommand\pmsn{\par\medskip\noindent}
\newcommand\pbsn{\par\bigskip\noindent}

\newlength{\algoindent}   
\newcommand\emdef[1]{\textsf{#1}}
\newcommand\emshort[1]{\emph{#1}}
\newcommand\emlong[1]{\emph{#1}}

\newcommand\red[1]{{\color{red} #1}}
\newcommand\redbf[1]{{\color{red} \textbf{#1}}}

\newcommand{\annotate}[1]{\if\MINE1\red{$!!!\bm\Rightarrow$} #1 \redbf{$\bm\Leftarrow!!!$}\fi}

\newcommand{\true}{\ensuremath{\mathsf{True}}\xspace}
\newcommand{\false}{\ensuremath{\mathsf{False}}\xspace}
\newcommand\tvt{\ensuremath{\mathsf{T}}\xspace}  
\newcommand\tvf{\ensuremath{\mathsf{F}}\xspace}  

\newcommand{\N}{\ensuremath{\mathbb{N}}\xspace}
\newcommand{\Ns}[1]{\ensuremath{\N_{#1}}\xspace}
\newcommand\none{\ensuremath{\bot}\xspace}
\newcommand{\prob}[1]{\mathrm{P}[#1]}

\newcommand{\al}[0]{\ensuremath{\Sigma}\xspace}   
   
\newcommand\als{\ensuremath{\al}\xspace}        
\newcommand{\ew}[0]{\ensuremath{\bm\varepsilon}\xspace}   
 
\newcommand\len[1]{|#1|}       
\newcommand\sz[0]{\mathsf{sz}}        
\newcommand\szinv[0]{\mathsf{sz}^{-1}}  

\newcommand\opfont{\mathcal}    
\newcommand{\langp}[1]{{\opfont L}(#1)}

\newcommand\lang{{\opfont L}}

\newcommand\aut{\ensuremath{\bm{a}}\xspace}   
\newcommand\autn{\ensuremath{\bm{n}}\xspace}   
\newcommand\autm{\ensuremath{\bm{m}}\xspace}   

\newcommand\dom{\ensuremath{\mathsf{dom}}\xspace}      

\newcommand{\err}[0]{\ensuremath{\varepsilon}\xspace}
\newcommand{\param}{\ensuremath{p}\xspace}  
\newcommand{\sspec}{\ensuremath{\alpha}\xspace}  

\newcommand{\cnt}{\ensuremath{\mathtt{Cnt}}\xspace}
\newcommand{\select}[0]{\stackrel{\$}{\longleftarrow}}
\newcommand{\NFA}[0]{\ensuremath{\mathsf{NFA}}\xspace}
\newcommand{\BNFA}[0]{\ensuremath{\mathsf{BNFA}}\xspace}
\newcommand{\dd}[0]{\ensuremath{D}\xspace} 
\newcommand{\dt}[0]{\ensuremath{T}\xspace} 
\newcommand{\ubfam}[0]{\ensuremath{\mathsf{B}}\xspace} 
\newcommand{\lbdu}[1]{\ensuremath{\langle#1\rangle}\xspace} 
\newcommand{\diri}[0]{\ensuremath{\mathsf{D}_{t,d}}\xspace}  

\newcommand{\ev}[1]{\ensuremath{\mathcal{E}(#1)}\xspace} 
\newcommand{\cond}[0]{\ensuremath{C}\xspace} 

\newcommand{\probd}[1]{\mathsf{prob}_{#1}} 
\newcommand\selectfin{\mathsf{selectFin}}
\newcommand\selectany{\mathsf{select}}
\newcommand\sample{\mathsf{SelectFinFam}}
\newcommand\szselect{\mathsf{sizeSelect}}
\newcommand\tosscoin{{\mathsf{tossCoin}}\xspace}
\newcommand\parestim{{\mathsf{ParamEstimate}}\xspace}
\newcommand\emptiness{{\mathsf{Emptiness}}\xspace}
\newcommand\universality{{\mathsf{Universality}}\xspace}
\newcommand\finemptiness{{\mathsf{FinEmptiness}}\xspace}
\newcommand\finuniversality{{\mathsf{FinUniversality}}\xspace}

\newcommand\maxlen{{\mathsf{maxLen}}\xspace}
\newcommand\distrparameter{{\mathsf{distrParameter}}\xspace}

\newcommand\classfont{\mathbf}
\newcommand\classP{\ensuremath{\classfont{P}}\xspace}

\newcommand\classcoNP{\ensuremath{\classfont{coNP}}\xspace}
\newcommand\PSPACE{\ensuremath{\classfont{PSPACE}}\xspace}

\newcommand\NL{\ensuremath{\classfont{NL}}\xspace}

\usepackage{bm}

\title{Improved Randomized Approximation of Hard Universality and Emptiness Problems\thanks{Research supported by NSERC, Canada (Discovery Grant of S.K.)}}

\author{
Pantelis Andreou\footnote{Department of Community Health and Epidemiology, Dalhousie University, Halifax, NS, Canada, \texttt{Pantelis.Andreou@dal.ca}}
\and
Stavros Konstantinidis\footnote{Department of Mathematics and Computing Science, Saint Mary's University, Halifax, NS, Canada, \texttt{s.konstantinidis@smu.ca}}
\and 
Taylor J.\ Smith\footnote{Department of Computer Science, St.\ Francis Xavier University, Antigonish, NS, Canada, \texttt{tjsmith@stfx.ca}}
}

\begin{document}

\maketitle

\begin{abstract}
We build on recent research on polynomial randomized approximation (PRAX) algorithms for the hard problems of NFA universality and NFA equivalence. 
Loosely speaking, PRAX algorithms use sampling of infinite domains within any desired accuracy $\delta$. 
In the spirit of experimental mathematics, we extend the concept of PRAX algorithms to be applicable to the emptiness and universality  problems in any domain whose instances  admit a tractable distribution as defined in this paper.
A  technical result here is that a linear  (w.r.t. $1/\delta$) number of samples  is sufficient, as opposed to the quadratic number of samples in previous papers. 
We show how the improved and generalized PRAX algorithms apply to universality and emptiness problems in various domains: ordinary automata, tautology testing of propositions, 2D automata, and to solution sets of certain Diophantine equations.
\end{abstract}

\section{Introduction}\label{sec:intro}
Polynomial randomized approximation (\emdef{PRAX}) algorithms were introduced in \cite{KoMaMoRo:2023} for deciding approximate versions of  NFA universality problems, as well as in \cite{KoMoReSe:2024} for deciding approximate versions of the NFA (in)equivalence problem.
The general idea, inspired from \cite[pg.~72]{MiUp:2017}, is to view estimating the fullness, or emptiness, of an NFA as the problem of estimating the parameter of some population and then following the tools of \cite{MiUp:2017} for parameter estimation problems. 
This idea can be applied to various types of objects (not only NFAs), where \emph{each object of interest describes a subset $L$ of some domain $X$}, and we wish to know whether $L$ is \err-close to being full, or \err-close to being empty, in $X$ for some desired tolerance $\err\in(0,1)$.
We are interested in particular in infinite domains $X$, where the subsets $L$  themselves can be infinite as well. 
Other than NFAs accepting  languages, we can talk about objects like expressions of  propositional logic where the domain is the set of truth assignments and the subset of interest is the set of satisfying truth assignments, or about Diophantine equations on three variables where the domain is the set of triples of integers and the subset of interest is the set of those triples satisfying the equation. 
As these problems can be hard (in the context of complexity theory or even of mathematics), we use sampling on the domain $X$ to obtain  information about the subset $L$, where the sampled information can be closer to the true answer as  \err gets smaller.
As there is no uniform distribution on countably infinite sets, we use sampling according to the Dirichlet distribution defined in \cite{Gol:1970,Gol:1992}.
\pnsi
In the spirit of experimental mathematics \cite{BorwDevl:2008}, we are interested in efficient versions of polynomial algorithms and their implementations that would allow us to obtain information about approximate solutions to instances of hard problems---in the same philosophy that \cite{CalDum:2020} attempts to solve instances of the halting~problem.

\pssn
\textbf{Structure of the paper and main results.}
The next section contains very basic notions and notation about formal languages, automata, hard problems in various domains, and probability distributions. We also define truncated distributions: finite versions of countably infinite distributions, where an infinite part of low probability is omitted. 
\underline{In Section}~\ref{sec:select}, we make precise how to algorithmically select (sample) elements from finite and infinite distributions using ideas from \cite{ArBa:2009,KoMaMoRo:2023}. For finite distribution families we use the concept of \emdef{polynomially samplable} distribution, as suggested in \cite{ArBa:2009}. For infinite distributions, we define \emdef{tractable} distributions \dt, where there is a polynomial algorithm $\selectany_{\dt}(\delta)$ that figures out how to cut an infinite ``tail'' of events such that the remaining finite set of events can be sampled as close to the infinite distribution as desired (i.e., w.r.t. a given tolerance $\delta$). 
We also define \emdef{locally tractable} distributions, which are applicable to a wide class of domains, and for which the main steps of the algorithm $\selectany_{\dt}(\delta)$ become concrete (Theorem~\ref{th:locally:is:tractable}).
\underline{In Section}~\ref{sec:bound}, we apply a certain upper bound of \cite{HeZhaZha:2010} for probabilities of the form $\prob{X\ge a}$ to show (Theorem~\ref{th:newbound}) that the number of samples required for emptiness and universality problems is linear with respect to the tolerance $\delta$, as opposed to the quadratic number in the general case \cite{MiUp:2017,KoMaMoRo:2023}. 
\underline{In Section}~\ref{sec:prax}, we define what we mean by  polynomial randomized approximation (\emdef{PRAX}) algorithms for any emptiness and any universality problems whose instances \sspec describe subsets $\lang(\sspec)$ of tractable domains, and we show that these PRAX algorithms do exist (Theorem~\ref{th:general:prax}). We do the same for the case where the domains involved are finite and polynomially samplable (Corollary~\ref{cor:finite:prax}). 
\underline{In Section}~\ref{sec:concrete}, we apply the general PRAX algorithms of Theorem~\ref{th:general:prax} and Corollary~\ref{cor:finite:prax} to concrete domains where we can also give time complexity estimates: tautology testing, universality of 2D automata, and emptiness of 3-variable Diophantine equations. Finally, \underline{Section}~\ref{sec:last} contains a few concluding remarks.

\section{Basic Notions and Notation}\label{sec:notation}
We use the notation \N for the set of  positive integers,  \Ns0 for the nonnegative integers, and $\N_0^{>x}$ (resp., $\N_0^{\le x}$) for the nonnegative integers greater than $x$ (resp., less than or equal to $x$), where $x$ is any real number.
If $S$ is a set, then $|S|$ denotes the cardinality of $S$.

We assume the reader to be familiar with basics of formal languages \cite{FLhandbookI,HoMoUl:2001}. 
Our arbitrary alphabet variable will be $\als$ of some cardinality $s\in\N$: $s=|\als|$. 
The \emdef{Boolean alphabet}\, $\{\tvt,\tvf\}$ consists of the two truth values $\tvt$ and $\tvf$.
Using an alphabet we can form words (strings) of any length. We have the following notation:
\pssi
\ew = empty word, \qquad\qquad\qquad $\als^*$ =  all words, \qquad $|w|$ = length of word $w$,
\pnsi
$\als^\ell$ = all words of length $\ell$,\qquad $\als^{\le\ell}$ = all words of length at most $\ell$,
\pmsn
and also similar notation like $\als^{<\ell}$ and $\als^{\ge\ell}$. We also use the following notation.
\pssi
$\NFA$ = all  NFAs (nondeterministic finite automata).
\pssi
$\BNFA$ = all block NFAs = NFAs accepting languages of a fixed word length.
\pssi
$|\aut|$ = size of  NFA \aut =  number of states plus  number of transitions in \aut.
\pssi$\langp{\aut}$ = the language accepted by the NFA 
\aut.
\pmsn

A familiarity with two-dimensional (2D) formal language and automata theory is also beneficial, and so here we review a number of fundamental concepts~\cite{GiammarresiRestivo19972DLanguages,Inoue19912DAutomataSurvey,Smith2019TwoDimensionalAutomata,Smi:2021}. 
An $m \times n$ \emdef{2D word} over an alphabet $\Sigma$ is an $m$-row-by-$n$-column array of symbols drawn from $\Sigma$. 
A \emdef{2D language} is naturally a set of 2D words. 
We have the following notation: 
\pssi
$\Sigma^{**}$ = all 2D words, \quad $\Sigma^{m \times n}$ = all 2D words of dimension $m \times n$,
\pnsi
$|z|_{\text{R}}$ = number of rows of a 2D word $z$, \quad $|z|_{\text{C}}$ = number of columns of a 2D word $z$.
\pmsn

A \emdef{2D automaton} is a finite automaton that takes as input 2D words. 
Just as a two-way finite automaton operating on strings can move its input head left and right, a 2D automaton operating on 2D words can move its input head in four directions: up, down, left, and right. 
The input word to a 2D automaton is surrounded by a special marker symbol \# that prevents the input head of the automaton from moving outside of the boundaries of the input word.

\pssn
\textbf{Subset descriptions.}
In the fields of algorithms and complexity, a subset of some domain (set) $X$ is described using a finite expression $\sspec$. We shall write \emlong{$\langp{\sspec}$ to denote the subset of $X$ described by $\sspec$}. 
Further below, we shall talk about the probability $\dd\big(\langp{\sspec}\big)$ of the subset $\langp{\sspec}$, relative to a distribution \dd.
\pssi\qquad
We shall  use $D(\sspec)$ as a  shorthand notation for $D\big(\langp{\sspec}\big)$.
\pmsn
{\small(In general, the domain $X$ depends on \sspec, i.e. $X=X_{\sspec}$, but we omit references to \sspec when there is no risk of confusion.)}

\begin{example}\label{ex:subset:descr} \textsf{[Subset Descriptions]}
Here are a few examples of subset descriptions \sspec with some associated hard problems about \sspec.
\begin{itemize}
    \setlength{\itemsep}{0pt}%
    \setlength{\parskip}{0pt}%
	\item $\sspec$ is a regular expression, or an NFA, over some alphabet $\als$ and $\langp{\sspec}$ is the language (subset of $\als^*$) described by $\sspec$. 
	A regular expression, or NFA, \sspec is \emdef{universal} if $\lang(\sspec)=\als^*$. Deciding whether $\lang(\sspec)=\als^*$ is a \PSPACE-hard problem. If the NFA \sspec is a block NFA of some word length $\ell$ then the problem of whether \sspec is block-universal ($\lang(\sspec)=\als^\ell$) is \classcoNP-hard \cite{KMR:2018}.
	\item $\sspec=\sspec(v_1,\ldots,v_k)$ is a proposition in conjunctive normal form (CNF) involving the truth variables $v_i$ and $\langp{\sspec}$ is the set of truth assignments $f=t_1\cdots t_k$, with each $t_i\in\{\tvt,\tvf\}$, that satisfy $\sspec$: $f(\sspec)=\tvt$. 
	For example the CNF proposition $(v_1\lor \bar v_2)\land(v_2\lor v_3)$ is satisfiable via the truth assignment $f=\tvt\tvt\tvf$, so $f\in\langp{\sspec}$. 
	A CNF proposition $\sspec$ is a \emdef{tautology} if every truth assignment  satisfies $\sspec$. Deciding whether \sspec is a tautology is a \classcoNP-hard problem \cite[pg.~56]{ArBa:2009}. 
	\item \sspec is a 2D automaton over some alphabet $\als$ and $\langp{\sspec}$ is the 2D language (i.e., a subset of $\als^{**}$) described by $\sspec$. The universality problem (whether $\lang(\sspec)=\als^{**}$) and the emptiness problem (whether $\lang(\sspec)=\emptyset$) are undecidable for 2D automata, but for restricted forms of 2D automata they can be \PSPACE-hard \cite{Smith2019TwoDimensionalAutomata,Smi:2021}.
	\item $\sspec=\sspec(x_1,\ldots,x_k)$ is a Diophantine polynomial expression and $\langp{\sspec}$ is the set of integer tuples $(n_1,\ldots,n_k)$ such that $\sspec(n_1,\ldots,n_k)=0$. 
	For example, the Diophantine polynomial expression $\sspec=x^2+y^2-z^2$ describes the set $\langp{\sspec}$ of all Pythagorean triples $(n_1,n_2,n_3)$ such that $n_1^2+n_2^2=n_3^2$. 
	On the other hand, $\langp{\sspec_k}=\emptyset$ when $\sspec_k=x^k+y^k-z^k$, for any fixed integer $k\ge3$.
	The general question of whether a given Diophantine equation has a solution is undecidable. However, even for specific simple looking Diophantine equations there is no information as to whether their set of  solutions is empty \cite{Grechuk:2022}.
\end{itemize}
\end{example}

\pssn\textbf{Probability distributions.}
We give here a quick presentation on probability distributions in the spirit of \cite{Gol:1970,Gol:1992,KoMaMoRo:2023}.
Our presentation improves the concept of ``augmented distribution'' in \cite{KoMaMoRo:2023}, which is intended to make smaller finite versions of (usually infinite) length distributions by cutting a set of events with low probability. 
Here we use  the term \emdef{truncated distribution} which is more general and can be used to make smaller finite versions of any distribution, not necessarily a length distribution.

\pnsi
Let $X$ be a countable set. A \emdef{probability distribution} on $X$ is a function $\dd:X\to[0,1]$ such that $\sum_{x\in X}\dd(x)=1$.
The \emdef{domain} of \dd, denoted by $\dom\dd$, is the subset $\{x\in X:\dd(x)>0\}$ of $X$. 
If $X$ is \emshort{finite}, that is, $X=\{x_1,\ldots,x_k\}$ for some $k\in\N$, then we write 
$$D=\big(D(x_1),\ldots,D(x_k)\big).$$
If $X\subseteq\N_0$ then the distribution $D$ is called a \emdef{length distribution}, and if $X\subseteq\N_0\times\N_0$ then the distribution $D$ is called a 2D length distribution---see below the Dirichlet length distributions  in Example~\ref{ex:dirichlet}.
We use the following terminology from \cite{Gol:1992}.

\begin{definition}\label{def:prob}
    Let $\dd$ be a probability distribution on $X$. For any subset $L$ of $X$, we define the quantity
\begin{equation}\label{eq:uindex}\dd(L)=\sum_{x\in L}\dd(x)\end{equation}
and refer to it as \emdef{the probability that a randomly selected element from \dd is in $L$.} The following notation, borrowed from cryptography, means that $x$ \emdef{is randomly selected from} $\dd$:\quad
$
x\select\dd.
$
\end{definition}

\begin{example}\label{ex:dirichlet}
\textsf{[Dirichlet distributions]} 
For any $t>1$ and $d\in\N_0$, the Dirichlet distribution $\diri$ on $\N_0$ is defined such that $\diri(\ell)=(1/\zeta(t))(\ell+1-d)^{-t}$ for $\ell\in\N_0^{\ge d}$, \cite{Gol:1970,KoMaMoRo:2023}, where $\zeta$ is the Riemann zeta function. 
This is a length distribution. The original version in \cite{Gol:1970} uses no parameter $d$ (i.e., $d=0$); however, as in \cite{KoMaMoRo:2023}, here we allow $d>0$ when it is desirable to skip small lengths (e.g., $\ell=0$) that would otherwise get a high probability. 
In \cite{Gol:1970}, the  author considers the Dirichlet distribution to be the basis  where  \emlong{``many heuristic probability arguments based on the fictitious uniform distribution on the positive integers become rigorous statements.''}
For some applications, we need 2D versions of $\diri$: \[
\diri^2(k,\ell)=\diri(k)\diri(\ell),
\]
which can be used to select two independent lengths. This is a 2D length distribution.
\end{example}

\pssn 
\textbf{Truncated distributions.}
Selecting from a distribution $D$ with a large domain $X=\dom D$ could return an element of large size, which can be intractable with respect to the algorithmic size of the input. 
For example, selecting from $D=\diri$ could return an arbitrarily large length, albeit with low probability.
For this reason, given a desirable ``small'' positive $\delta<1$, we can choose a finite subset $F$ of $X$ and then we select elements from $F$, omitting the (usually infinite) \emdef{tail} $X-F$ of $D$, provided that $D(X-F)\le\delta$. 
More specifically, if $F=\{x_1,\ldots,x_k\}$, the \emdef{$F$-truncated version of $D$} is 
the distribution $D^F$ with domain $F\cup\{\none\}$, where `\none' is an object outside of $X$, such that
\[
D^F = \big(D(x_1),\ldots,D(x_k),\> 1-D(F)\big)
\]
Each element $x\in X$ has an algorithmic size 
and we normally choose $F$ such that the sizes of its elements are not too large w.r.t. the input size of the algorithm in which the selection takes place. 
Thus, the distribution $D^F$ can select the outcome `\none' instead of a large $X$-element.

\begin{example}\label{ex:truncated}
	The main example of a truncated length distribution $D$ in \cite{KoMaMoRo:2023} is the distribution $D^F$ where $F=\N_0^{\le M}$, for some desirable maximum length $M$.
	This distribution truncates all lengths $>M$.
	In fact, in \cite{KoMaMoRo:2023}, $D^F$ is denoted by $D^M$ and is called an ``augmented distribution''.
\end{example}

\begin{example}\label{ex:word:distr} 
\textsf{[Word distributions]}
The distribution $\diri$ can be used as a first step to select a word in $\als^*$: first, select a length $\ell$ with probability $\diri(\ell)$, and then select uniformly from \als each of the $\ell$ symbols of the word. 
This process defines the \emdef{Dirichlet word distribution} $\lbdu{\diri}$ on $\als^*$ such that $$\lbdu{\diri}(w) = \diri(\len w)\cdot s^{-\len w}\quad \text{for all } w\in\als^*$$
where $s=|\als|$.
We have that $\lbdu{\diri}(\als^\ell)=\diri(\ell)$ for any word length $\ell$.
Similarly, the distribution $\diri^2$ can be used as a first step to select 2D words:  first, select the number of rows $k$ and the number of columns $\ell$, and then select uniformly from \als each of the $k\times\ell$ alphabet symbols of the 2D word. 
This process defines the \emdef{2D Dirichlet word distribution} $\lbdu{\diri^2}$ on $\als^{**}$ such that 
$$\lbdu{\diri^2}(z) = \diri^2(|z|_{\mathrm R},|z|_{\mathrm C})\cdot s^{-(|z|_{\mathrm R}+|z|_{\mathrm C})}\quad \text{for all } z\in\als^{**}.
$$
We have that $\lbdu{\diri^2}(\als^{k\times\ell})=\diri^2(k,\ell)$ for any word dimensions $k,\ell$.
\end{example}

\emph{\textbf{Note:} Following \cite[pg.~126]{ArBa:2009} and \cite{KoMaMoRo:2023}, we shall assume unit cost for arithmetic operations.}
\pssi
We shall use a few times the following version of a result from \cite[Lemma~6]{KoMaMoRo:2023}.

\begin{lemma}\label{lem:dirichlet}
Let $\delta\in(0,1)$ and $M\in\N$.
If $M\ge\sqrt[t-1]{1/\delta}+(d-1)$ then $\diri(\N_0^{>M})\le\delta$.
\end{lemma}

\section{How to Randomly Select from a Distribution}\label{sec:select}
Returning to the topic of subset descriptions \sspec, we wish to know whether $\sspec$ satisfies a certain property; for example, whether \sspec is universal: $\lang(\sspec)=X$. 
As these questions can be hard, we wish to get information about the problem instance \sspec by sampling elements of $X$.
So we require that $X$ = the domain of some distribution \dd  for which there is a polynomial (randomized) algorithm 
that returns  a randomly selected element from \dd. 
\pnsi
Following \cite{KoMaMoRo:2023}, for finite distributions $\dd=\big(\dd(x_1),\ldots,\dd(x_k)\big)$ on some set $\{x_1,\ldots,x_k\}$, we shall assume available the  (randomized) algorithm 	$\selectfin(\dd)$, which  returns a randomly selected $x_i$ with probability $\dd(x_i)$. 
The algorithm works in time $O(k)$ using the reasonable assumption of constant cost of $\tosscoin(p)$ and of arithmetic operations \cite[pg.~126, 134]{ArBa:2009}. 
Here, $\tosscoin(p)$ returns 0 or 1 with probability $p$ or $1-p$, respectively, where $p\in[0,1]$.
\pnsi
For a finite distribution $\dd$ that has an implicit description, it might not be efficient to list the value  $\dd(x)$ for every $x\in\dom\dd$.
 For example, if \sspec is a CNF proposition with some $k$ variables, then the sampling domain  $\{\tvt,\tvf\}^k$ = ``all truth assignments of length $k$'' is finite but exponentially large w.r.t. the size of \sspec. Despite this, we can sample in time $O(k)$ a truth assignment of length $k$ by simply invoking $k$ times the function $\tosscoin(1/2)$. 
 In this example, the distribution is implicitly described by the parameter $k$ that can be computed from \sspec. So in fact, following the approach of \cite[pg.~365]{ArBa:2009}, we consider a \emph{family} $\dd=(\dd_k)$ of finite distributions (which is associated to the instances \sspec of a desirable decision problem).
 
\begin{definition}\label{def:samplable}
A family $\dd=(\dd_k)$ of finite distributions is called \emdef{polynomially samplable} if there is a (randomized) algorithm $\sample(k)$ that selects an element from $\dd_k$ in polynomial time with respect to $k$, where $k$ is given in unary.
\end{definition}
\pnsi
Next we develop a method of selecting elements from a distribution \dt that has an \emph{infinite domain} $X=\dom\dt$. 
In reality, our method selects elements from a truncated version of \dt.
In the literature (e.g., \cite{BCGL:1992}), one does find the concept of polynomially samplable distribution even for distributions with infinite domain, but these relate to cryptographic models where arithmetic operations have non-constant complexities. 
Here, however, we aim at efficient (both theoretically and practically) sampling methods.
\pnsi
The simplest definition is the following.

\begin{definition}\label{def:tractable}
A distribution \dt with domain $X$ is called \emdef{tractable} 
if there is a  (randomized) algorithm $\selectany_{\dt}(\delta)$, where $\delta\in(0,1)$, that is polynomial w.r.t. $1/\delta$ and randomly selects an element from the truncated distribution $\dt^F$, where $F$ is a set of $X$-elements determined by the algorithm once (that is, $F$ is specified by a static variable) such that $\dt(X-F)\le\delta$.
\end{definition}

The above definition is meant for distributions with infinite domain. The algorithm $\selectany_{\dt}(\delta)$ returns either an element in $F\subseteq X$, or `\none' (indicating an attempt to select an improbable element). 
This definition, however, is too abstract and does not reveal any of the mechanics of the algorithm  $\selectany_{\dt}(\delta)$.  
Next we work toward a more concrete definition.
\pnsi
\emph{We assume that the domain $X$ of the distribution has a \emdef{size function} $\sz: X\to \N_0$ such that, for each $m\in\N_0$, the set $\szinv(m)=\{x\in X:\sz(x)=m\}$ is finite.}
\pssn 
\textbf{Conditional distributions and two-level sampling.}
Let $Y$ be a subset of the domain $X=\dom \dd$ of some distribution \dd. 
The conditional version of \dd on $Y$ is the distribution $(\dd|Y)$ with domain $Y$ such that $(\dd|Y)(y)=\dd(y)/\dd(Y)$.
The concept of conditional distribution can be helpful when we want to select from a distribution \dd with large domain:
first, select a subset $Y$ of $X$, and then select an element from $(\dd|Y)$. This two-level sampling returns $y\in Y$ with probability $\dd(Y)\cdot (\dd|Y)(y)=\dd(y)$.
We did this two-level sampling in Example~\ref{ex:word:distr} where we selected a word from $\lbdu{\diri}$ by first selecting a length $\ell$ (that is, the subset $\als^\ell$ of $\als^*$) and then selecting a word of length $\ell$.



\begin{definition}\label{def:locally:tractable}
	A  distribution \dt  whose domain $X=\dom\dt$ has a size function $\sz$ 
	is called \emdef{locally tractable} if the following conditions hold true.
\begin{enumerate}
    \setlength{\itemsep}{0pt}%
    \setlength{\parskip}{0pt}%
  \item There is a  randomized algorithm $\szselect_{\dt}(m)$, where $m\in\N_0$, that selects an element from $\big(\dt|\szinv(m)\big)$ and works in  polynomial time w.r.t. $m$. Thus, $\szselect_{\dt}(m)$ selects an $X$-element of size $m$ from \dt.
  \item There is an algorithm $\probd{\dt}(m)$, where $m\in\N_0$, that returns the value $\dt\big(\szinv(m)\big)$ and works in polynomial time w.r.t. $m$. Thus, $\probd{\dt}(m)$ returns the probability of the set of $X$-elements of size $m$.
  \item For all $\delta\in(0,1)$, there is $M\in\N_0$ such that $\dt\big(\szinv(\N_0^{>M})\big)\le\delta$, $M$ is of polynomially bounded magnitude  w.r.t. $1/\delta$ (that is, 
  $M=O\big((1/\delta)^k\big)$ for some $k\in\N_0$), and there is an algorithm $\maxlen_\dt(\delta)$ that returns such an $M$ and works in polynomial time w.r.t. $1/\delta$. 
  Thus, $M$ is intended to define the truncated version $\dt^F$ of $\dt$, where $F=\{x\in X:\sz(x)\le M\}$.
\end{enumerate}
\end{definition}

\pnsn
Definition~\ref{def:locally:tractable}  gives a broad class of distributions for which sampling can be done in polynomial time and is applicable to many problem domains.
Theorem~\ref{th:locally:is:tractable} below shows how the definition can be used to write the desired algorithm $\selectany_{\dt}(\delta)$ in Definition~\ref{def:tractable} that selects an element from a truncated version of \dt.

\begin{example}\label{ex:locally:tractable}
We consider the domain $\al_2^*$  of binary words and the Dirichlet word distribution $T=\lbdu{\diri}$ of Example~\ref{ex:word:distr}. 
We use the size function $\sz(w)$ = the length of the word $w$, so $\szinv(m)=\al_2^m$.
Now, we explain that $T=\lbdu{\diri}$ is locally tractable:
(i) The required algorithm $\szselect_{\dt}(m)$ uniformly selects a binary word of length $m$ by using $m$ times the function call $\tosscoin(1/2)$; and this works in time $O(m)$.
(ii) The required algorithm $\probd{T}(m)$ simply computes  $T(\al_2^m)=\diri(m)=1/\big(\zeta_t\cdot (1+m)^t\big)$,  which can be done in constant time, assuming unit cost for arithmetic operations \cite[pg.~126]{ArBa:2009}. 
(iii) By Lemma~\ref{lem:dirichlet}, if $M\ge \sqrt[t-1]{1/\delta}+d-1$, then $\diri(\N_0^{>M})\le\delta$, which is equivalent to $T(\al_2^{>M})\le\delta$; hence, $M$ can be computed in constant time, assuming again unit cost for arithmetic operations.
\end{example}

\begin{theorem}\label{th:locally:is:tractable}
Every locally tractable  distribution is tractable.
\end{theorem}

\begin{proof} 
We show the steps of the algorithm $\selectany_{\dt}(\delta)$ specified in Definition~\ref{def:tractable}.
The first step is to invoke $\maxlen_{\dt}(\delta)$ to get $M$ = the maximum size of the $X$-elements that will be selected. 
Let $F=\szinv(\N_0^{\le M})=\{x\in X: \sz(x)\le M\}$.
The second step is to use $\probd{\dt}(m)$, for $m=0,\ldots, M$, to compute the distribution 
\[
P_M = \Big(T(\szinv(0)), T(\szinv(1)),\ldots, T(\szinv(M)),\>1-\sum_{i=0}^MT(\szinv(i))\Big),
\]
The third step is to select from $P_M$ a size $m$ or `\none'. 
The last step is to either return \none or use $\szselect_{\dt}(m)$ to return an element of $F$.
Note that the first two steps are used only in the first invocation of the algorithm where $P_M$ is initialized.
For any $x\in F$, let $P(x)$ be the probability that the above algorithm returns $x$. Let $m=\sz(x)$. We confirm that $P(x)=\dt^F(x)$:
\[
P(x)= P_M(m)\cdot \big(T|\szinv(m)\big)(x) = P_M(m)\cdot \frac{T(x)}{T(\szinv(m))}=T(x)=T^F(x). \qedhere
\]
\end{proof}

\section{Improved Sample Size Bound for Parameter Estimate}\label{sec:bound}
A main technical contribution of this work is Theorem~\ref{th:newbound}, which implies that a linear sample size $n=n(\delta)$ is sufficient to use in a randomized approximation algorithm (with tolerance $\delta$) for both emptiness and universality problems, as opposed to the quadratic sample sizes in \cite{KoMaMoRo:2023} and \cite{KoMoReSe:2024}.
More specifically, if the condition $C$ in the below theorem specifies a subset $L$ of $X=\dom D$ and $p=D(L)$, then (i) if $p<1-\delta$ (i.e., $L$ is not close to being universal) then it is not probable that all $n$ sampled $X$-elements belong to $L$; (ii) if $p>\delta$ (i.e., $L$ is not close to being empty) then it is not probable that none of the $n$ sampled $X$-element belongs to $L$.
In either case, the linear sample size $n$ suffices.

\begin{theorem}\label{th:newbound}
Let \cnt be the random variable for the final value of $\mathrm{cnt}$ in the random process of Fig.~\ref{fig:param:est}. 
Let \param =  the probability that $x$ satisfies condition  $\cond$ when $x$ is selected from the distribution $D$. Let $c=4.76603$.
\begin{itemize}
	\item If $p<1-\delta$ and $n\ge c/\delta$ 
	then $\prob{\cnt=n}\le1/4$.
	\item If $p>\delta$ and $n\ge c/\delta$ then $\prob{\cnt=0}\le1/4$.
\end{itemize}
\end{theorem}	

\begin{figure}[t]
\begin{center}
\parbox{0.40\textwidth}
{
\hspace*{0.3\algoindent} 
$\parestim\,(n,D,C)$
\pssn \hspace*{1.5\algoindent}
cnt := 0;
\\ \hspace*{1.5\algoindent}
repeat $n$ times:   
\\\hspace*{2.5\algoindent} 
$x$ $\select$ $D$;
\\ \hspace*{2.5\algoindent} 
if ($x$ satisfies $C$) 
\\ \hspace*{3.5\algoindent}
cnt := cnt+1;
\\ \hspace*{1.5\algoindent} 
return cnt / $n$;
}
\caption{
This random process returns an estimate of\, \param =  \emph{the probability that an element $x$ selected from the distribution  $D$ satisfies condition $C$}. 
The parameter $n$ is the number of samples to select from the  distribution $D$.  
Example 1: The condition $C$ is whether a truth assignment $x$ satisfies a certain CNF proposition \sspec. Hence, if \sspec involves some $k$ variables then  $D$ = the uniform distribution on $\{\tvt,\tvf\}^k$.  This process can be used to determine whether the proposition \sspec is close to being a tautology. 
Example 2: The condition $C$ refers to some  
NFA \sspec over some alphabet \al. The distribution is $D = \lbdu{\diri^F}$ = the word distribution based on a truncated  Dirichlet distribution. 
The condition is whether ``$x=\none$ or $x\in\lang(\sspec)$''. This condition is used in \cite{KoMaMoRo:2023} with regards to how close \sspec is universal relative to~$\lbdu{\diri}$.
}\label{fig:param:est}
\end{center}
\end{figure}

\begin{remark}\label{rem:new:bound}
The sample size bound for parameter estimates used in 
	\cite{KoMaMoRo:2023,KoMoReSe:2024}  is based on the Chebyshev bound (inequality), which becomes as follows when phrased  for a binomial random variable $B$, where $a>0$:
\begin{equation}\label{eq:cheb}
	\prob{\,|B-E(B)|\ge a} \;\le\; n/(4a^2).
\end{equation}
We also note that using the  Chernoff bound, \cite[Equation~(4.6)]{MiUp:2017},
\[
\prob{\,B-np\ge\delta np}\le 2 e^{-np\delta^2/3}
\] 
with $p<1-\delta$, again leads to a quadratic value for $n$.
\end{remark}

\pnsi
For the proof of the above theorem,
we use as a lemma the following result from \cite[Case 2 of Theorem~2.3]{HeZhaZha:2010}.
\begin{lemma}\label{lem:known}
Let $a>0$, let $X$ be a random variable with $\ev{X}=0$, let
	$K=M_4/M_2^2$, and let $L=M_2/a^2$, where $M_2,M_4$ are the second and fourth moments of $X$. If $K\le L+\frac1L-1$ and $L<1$ then 
\[
\prob{X\ge a}\le 
\frac{M_4-M_2^2}{M_4-2M_2a^2+a^4}.
\]
\end{lemma}

\pmsn
\begin{proof} (Of Theorem~\ref{th:newbound})
First, note that \cnt is binomial: the number of successes  in $n$ selections (when the selected $x$'s satisfy the condition).
Hence $\ev{\cnt}=np$.
Then, for $X=\cnt-np$ or $X=np-\cnt$, we have that $\ev{X}=0$. Moreover, for $q=1-p$, we have
\[
M_2=npq \quad \text{and} \quad M_4=npq+npq(3n-6)pq.
\]
\underline{First statement}: $p<1-\delta$.
As $\cnt\le n$ always holds, we have $\prob{\cnt=n}=\prob{\cnt\ge n}=\prob{\cnt-np\ge n-np}=\prob{X\ge nq}$, where $X=\cnt-np$. Let $a=nq$.
With the notation in Lemma~\ref{lem:known}, we have that 
\[
K=\frac{3pqn+1-6pq}{npq},\quad 
L=\frac{p}{nq},\quad
L+\frac1L-1=\frac{p^2+n^2q^2-npq}{npq}.
\]
As $n>1/\delta$ and $q>\delta$, we have that $L<1$. 
We show now that $K\le L+1/L-1$ so that we can apply Lemma~\ref{lem:known}. Indeed $K\le L+1/L-1$ is equivalent to 
\begin{equation}\label{eq:quadratic}
n^2q-4n(1-q)+4-5q\ge0.	
\end{equation}
Using $n\to\infty$ as variable, the discriminant is $4(4-12q+9q^2)=36(q-2/3)^2\ge0$, hence, \eqref{eq:quadratic} is equivalent to $n\ge \big(2-2q+3|q-2/3|\big)/q$, which holds true as $n>4/\delta$, $1/\delta>1/q$, and $2/3>|q-2/3|$ for all $q>0$.
Lemma~\ref{lem:known} implies that 
\[
\prob{\cnt=n}=\prob{X\ge a}\le \frac{M_4-M_2^2}{M_4-2M_2a^2+a^4}.
\]
We need to show that the above fraction is $\le 1/4$. One confirms that this is the case iff 
\[
f(n) = n^3q^3-2n^2(q^2-q^3)-5n(q+q^3-2q^2)-3+21q-36q^2+18q^3\ge0.
\]
We have that $\lim_{n\to-\infty}f(n)=-\infty$ and $\lim_{n\to\infty}f(n)=\infty$ and that $f(n)$ has a local maximum $n_1$ and a local minimum $n_2>n_1$ with
\[
n_1,n_2 = \frac{(1-q)(2\pm\sqrt{19})}{3q}.
\]
The local maximum and minimum values are the zeroes of $f'(n)$, and $n_2$ is the local minimum as $f''(n_2)>0$. Hence, $f(n)$ is increasing for $n\ge n_2$; so it is sufficient to show that $f(n)>0$ for $n= c/q$ (which would imply that $f(n)>0$ for all $n\ge c/\delta$). 
We have that $f(c/q)>0$, if $c^3\ge 2c^2+5c+39$ and this latter condition is satisfied for $c=4.76603$ as required---the condition is not satisfied for $c=4.76602$.
\pssn
\underline{Second statement}: $p>\delta$. This is symmetric to the previous case when we note the following
\begin{itemize}
    \setlength{\itemsep}{0pt}%
    \setlength{\parskip}{0pt}%
    \vspace{-0.7\topsep}
	\item As $\cnt\ge0$ always holds, we have $\prob{\cnt=0}=\prob{\cnt\le0}=\prob{np-\cnt\ge np}=\prob{X\ge np}$, where $X=np-\cnt$. 
	\item Let $a=np$. 
	\item We have $p>\delta$, whereas before we had that $q>\delta$. Thus the required inequality $\prob{\cnt=0}\le 1/4$ follows if in the  proof of the first statement we switch the roles of $p$ and $q$.
\qedhere
\end{itemize}
\end{proof}

\section{PRAX Algorithms for Emptiness and Universality}\label{sec:prax}
In this section, we derive PRAX algorithms that can be applied to hard universality and emptiness problems in any domain that has a tractable distribution (Theorem~\ref{th:general:prax}) or a polynomially samplable finite distribution  (Corollary~\ref{cor:finite:prax}).
\pnsi
\emph{We consider problems (formal languages) in which every instance   
\sspec specifies a domain $X=X_{\sspec}$ and a subset description; that is, each $\sspec$ describes a subset $\lang(\sspec)$ of $X$. We assume that  \emph{$\lang(\sspec)$ is polynomially decidable in $X$}. 
\pnsi
Each decision problem and each PRAX algorithm refers to a specific \emph{family} \dt of  tractable distributions.  Each problem instance \sspec implies a particular  tractable distribution $\dt_{\sspec}\in \dt$.
However, for the sake of notational simplicity, 
we write $\dt(\sspec)$ to mean $\dt_{\sspec}\big(\lang(\sspec)\big)$.
}
\pnsi
As an example, $\lbdu{\diri}$ is in fact a family of Dirichlet word distributions: if \autn is an NFA over some alphabet \als then the expression $\lbdu{\diri}(\autn)$ implies the Dirichlet word distribution on $\als^*$ = the words over the alphabet of \autn.
As a second example, $\ubfam$ is the family of \emdef{uniform block distributions}: if \autn is a block NFA of some word length $\ell$ over some alphabet \als, then the expression $\ubfam(\autn)$ implies the uniform distribution on $\als^\ell$.

\pssn\textbf{Universality problems.}
A \emdef{universality problem} relative to some tractable distribution family $\dt$ is a language\footnote{Following the presentation style of \cite[pg.~193]{Gold:2008} and \cite{KoMaMoRo:2023}, we refrain from cluttering the notation with the use of a variable for the set of instances.}
\[
U_\dt=\{\sspec : \dt(\sspec)=1\}.
\]
For example, in the NFA universality problem, each instance \autn is an NFA and the question is whether $\lang(\autn)=\als^*$, which is equivalent to whether $\dt(\autn)=1$. 
Each real $\err\in(0,1)$ defines the \emdef{approximation language}
\[
U_{\dt,\err}=\{\sspec :\dt(\sspec)\ge1-\err\}.
\] 
The idea here is that, as it is hard to tell whether $\dt(\sspec)=1$, we might be \emlong{happy to know that $\dt(\sspec)\ge1-\err$,} where \err is called the (approximation) \emdef{tolerance}. 
As $U_{\dt,\err}$ can be harder than $U_{\dt}$ \cite{KoMaMoRo:2023}, we define a 
\emdef{PRAX algorithm  for $U_{\dt}$} to be a randomized decision algorithm\footnote{A \underline{decision} algorithm halts on every input with the answer \true or \false. }
$A(\sspec,\err)$ satisfying the following conditions:   
\begin{enumerate}
    \setlength{\itemsep}{1pt}%
    \setlength{\parskip}{0pt}%
    \item if $\sspec\in U_{\dt}$ then $A(\sspec,\err)=\true$;
    \item if $\sspec\notin U_{\dt,\err}$ then $\prob{A(\sspec,\err)=\false}\ge3/4$;
    \item $A(\sspec,\err)$ works in polynomial time w.r.t. $1/\err$ and the size of $\sspec$.
\end{enumerate}
When $A(\sspec,\err)$ gives the answer \false, this answer is correct: $\sspec\notin U_\dt$.
If $A(\sspec,\err)$ returns \true then probably $\sspec\in U_{\dt,\err}$, in the sense that $\sspec\notin U_{\dt,\err}$ would imply $\prob{A(\sspec,\err)=\false}\ge 3/4$.
Thus, when the algorithm returns \true, the answer is \emlong{correct within the tolerance $\err$} (i.e., $\sspec\in U_{\dt,\err}$) with  probability $\ge 3/4$. 
The algorithm returns the wrong answer exactly when it returns \true and $\sspec\notin U_{\dt,\err}$, but this happens with probability $\le1/4$.

\pmsn
\textbf{Emptiness problems.}
An \emdef{emptiness problem} relative to  some tractable distribution family  $\dt$ is a language 
\[
E_\dt=\{\sspec : \dt(\sspec)=0\}.
\]
For example, in the NFA equivalence problem, each instance is a pair $(\autm,\autn)$ of NFAs and the question is whether 
$\lang(\autm)\triangle\lang(\autn) =\emptyset$, 
which is equivalent to whether $\dt\big(\lang(\autm)\triangle\lang(\autn)\big)=0$. 
As before, each tolerance $\err\in(0,1)$ defines an
approximation  language 
\[
E_{\dt,\err}=\{\sspec :\dt(\sspec)\le \err\}.
\] 
Thus, $E_{\dt,\err}$ consists of instances for which the subset $\lang(\sspec)$  is very small, so even when a randomized algorithm detects no element in $\lang(\sspec)$, we are happy to accept that $\lang(\sspec)$ is close to empty.
We define a \emdef{PRAX algorithm for $E_\dt$}  to be a randomized decision algorithm $A(\sspec,\err)$ such that 
\begin{enumerate}
    \setlength{\itemsep}{1pt}%
    \setlength{\parskip}{0pt}%
	\item If $\sspec\in E_{\dt}$ then $A(\sspec,\err)=\true$.
	\item If $\sspec\notin E_{\dt,\err}$ then $\prob{A(\sspec,\err)=\false}\ge 3/4$.
	\item $A(\sspec,\err)$ works in polynomial time w.r.t. $1/\err$ and the size of $\sspec$.
\end{enumerate}
When $A(\sspec,\err)$ gives the answer \false, this answer is correct: 
	$\sspec\notin E_{\dt}$. 
If $A(\sspec,\err)$ returns \true then probably $\sspec\in E_{\dt,\err}$, in the sense that $\sspec\notin E_{\dt,\err}$ would imply $\prob{A(\sspec,\err)=\false}\ge 3/4$. 
Thus, whenever the algorithm returns the answer \false, this answer is correct: $\sspec\notin E_{\dt}$; when the algorithm returns \true, the answer is \emlong{correct within the tolerance $\err$} (i.e., $\sspec\in E_{\dt,\err}$) with  probability $\ge 3/4$. 
The algorithm returns the wrong answer exactly when it returns \true and $\sspec\notin E_{\dt,\err}$, but this happens with probability $\le1/4$.

%

\begin{figure}[t]
\begin{center}
\parbox{0.42\textwidth}
{
$\emptiness_{\dt}\,(\sspec,\err)$
\pnsn \hspace*{1\algoindent}
$c := 4.76603$;
\pnsn \hspace*{1\algoindent}
$n := \lceil c/(\err/2)\rceil$;
\pnsn \hspace*{1\algoindent}
$M := \maxlen_{\dt}(\err/2)$;
\pnsn \hspace*{1\algoindent}
for $i := 0,\ldots,M$
\pnsn \hspace*{2.2\algoindent}
$t_i := \probd{\dt}(i)$;
\pnsn \hspace*{1\algoindent}
$r := 1-\sum_{i=0}^{M} t_i$;
\pnsn \hspace*{1\algoindent}
$D=\big(t_0,\ldots,t_M,r\big)$;
\\ \hspace*{1\algoindent}
repeat $n$ times:   
\\\hspace*{2.2\algoindent} 
$\ell := \selectfin(D)$;
\\ \hspace*{2.2\algoindent} 
if ($\ell\neq\none$) 
\\ \hspace*{3\algoindent}
$x := \szselect_{\dt}(\ell)$;
\\ \hspace*{3\algoindent}
if $\big(x\in\lang(\sspec)\big)$ 
\\ \hspace*{4\algoindent}
return \false
\\ \hspace*{1\algoindent} 
return \true;
}
\quad
\parbox{0.42\textwidth}
{
$\universality_{\dt}\,(\sspec,\err)$
\pnsn \hspace*{1\algoindent}
$c := 4.76603$;
\pnsn \hspace*{1\algoindent}
$n := \lceil c/(\err/2)\rceil$;
\pnsn \hspace*{1\algoindent}
$M := \maxlen_{\dt}(\err/2)$;
\pnsn \hspace*{1\algoindent}
for $i := 0,\ldots,M$
\pnsn \hspace*{2.2\algoindent}
$t_i := \probd{\dt}(i)$;
\pnsn \hspace*{1\algoindent}
$r := 1-\sum_{i=0}^{M} t_i$;
\pnsn \hspace*{1\algoindent}
$D=\big(t_0,\ldots,t_M,r\big)$;
\\ \hspace*{1\algoindent}
repeat $n$ times:   
\\\hspace*{2.2\algoindent} 
$\ell := \selectfin(D)$;
\\ \hspace*{2.2\algoindent} 
if ($\ell\neq\none$) 
\\ \hspace*{3\algoindent}
$x := \szselect_{\dt}(\ell)$;
\\ \hspace*{3\algoindent}
if $\big(x\notin\lang(\sspec)\big)$ 
\\ \hspace*{4\algoindent}
return \false
\\ \hspace*{1\algoindent} 
return \true;
}
%
\caption{
PRAX algorithms for the emptiness problem $E_{\dt}$ (on the left) and for the universality problem $U_{\dt}$ (on the right), where \dt is a family of locally tractable distributions. The time complexity is $O$-bounded by $\>M\cdot\mathrm{Cost}\big(\probd{\dt}(M)\big)+
(1/\err)\cdot\big(M+\mathrm{Cost}\big(\szselect_{\dt}(M)\big)+\mathrm{Cost}\big(x\in\lang(\sspec)\big)\big)$.
}\label{fig:prax}
\end{center}
\end{figure}

\begin{theorem}\label{th:general:prax}
Let \dt be a locally tractable distribution family.
Let $E_{\dt}$ be any emptiness problem  and let $U_{\dt}$ be any universality problem, where their instances \sspec are subset descriptions with  polynomially decidable $\lang(\sspec)$.
The two algorithms in Fig.~\ref{th:general:prax} are PRAX algorithms  (relative to \dt) for $E_{\dt}$ and $U_{\dt}$.
\end{theorem}

\begin{proof}
First, we confirm that both algorithms work in time polynomial w.r.t. to $1/\err$ and  the size of \sspec: 
(i) as \dt is locally tractable, the functions $\maxlen_{\dt}$ and $\probd{\dt}$ work in polynomial time, and the value of $M$ is polynomially large; 
(ii) as \dt is locally tractable, the function $\szselect(\ell)$ works in polynomial time; 
(iii) as the subset $\lang(\sspec)\subseteq X$ is polynomially decidable in $X$, the tests $x\in\lang(\sspec)$ and $x\notin\lang(\sspec)$ can be done in polynomial time. 
\pnsi
As $X$ is the domain of a locally tractable distribution, there is a size function $\sz:X\to \N_0$.
Let $F$ be the finite subset $\szinv(\N_0^{\le M})$ of $X$; that is, all elements of $X$ whose size is $\le M$. 
Then, $X-F=\szinv(\N_0^{>M})$.
By definition of $\maxlen_{\dt}$, we have that $\dt(X-F)\le \err/2$.
\pnsi
For brevity, we use $A(\sspec,\err)$ for  $\emptiness_{\dt}(\sspec,\err)$. We show that $A(\sspec,\err)$ is a PRAX algorithm for $E_{\dt}$. 
Let \sspec be an instance of $E_{\dt}$. If $\sspec\in E_{\dt}$ then $\dt(\sspec)=0$, hence $\lang(\sspec)=\emptyset$ and the algorithm must return \true.
If $\sspec\notin E_{\dt,\err}$, then $\dt(\sspec)>\err$. We use Theorem~\ref{th:newbound} to show that $\prob{A(\sspec,\err)=\true}\le 1/4$. 
First, note that the algorithm $A(\sspec,\err)$ is logically equivalent to the version of the random process in Fig.~\ref{fig:param:est} where (i) the condition $\cond$ is ``$x\neq\none$ and $x\in\lang(\sspec)$'', which is equivalent to ``$x\in\lang(\sspec)\cap F$''; and (ii) instead of $\mathrm{cnt}/n$ the process returns \true if cnt = 0 and \false otherwise. Hence, the parameter $p$ in Theorem~\ref{th:newbound}~is 
$$
p
=\dt\big(\lang(\sspec)\cap F\big) 
\ge T\big(\lang(\sspec))- T\big(X-F\big)>\err-\err/2=\err/2.
$$
Hence, Theorem~\ref{th:newbound} implies that $\prob{A(\sspec,\err)=\true}\le 1/4$, as required.
\pnsi
For brevity, we use $B(\sspec,\err)$ for  $\universality_{\dt}(\sspec,\err)$. We show that $B(\sspec,\err)$ is a PRAX algorithm for $U_{\dt}$. 
Let \sspec be an instance of $U_{\dt}$. If $\sspec\in U_{\dt}$ then $\dt(\sspec)=1$, hence $\lang(\sspec)=X$ and the algorithm must return \true.
If $\sspec\notin U_{\dt,\err}$, then $\dt(\sspec)<1-\err$. We use Theorem~\ref{th:newbound} to show that $\prob{B(\sspec,\err)=\true}\le 1/4$. 
Note that the algorithm $B(\sspec,\err)$ is logically equivalent to the version of the random process in Fig.~\ref{fig:param:est} where (i) the condition $\cond$ is ``$x=\none$ or $x\in\lang(\sspec)\cap F$''; and (ii) instead of $\mathrm{cnt}/n$ the process returns \true if cnt = $n$ and \false otherwise. Hence, the parameter $p$ in Theorem~\ref{th:newbound} is 
$$
p
=\dt\big(X-F) + \dt\big(\lang(\sspec)\cap F\big)
\le  T\big(X-F\big) + T\big(\lang(\sspec))  
< \err/2+1-\err=1-\err/2.
$$
Hence, Theorem~\ref{th:newbound} implies that $\prob{B(\sspec,\err)=\true}\le 1/4$, as required.
\end{proof}

\begin{remark}\label{rem:choice:of:M}
The choice of the argument $\err/2$ in the function call $\maxlen_{\dt}(\err/2)$ is not insignificant as it affects the magnitude  of $nM$, which is a factor of the time complexity of the algorithm.
More specifically, consider using $M=\maxlen_{\dt}(\delta)$
and $n=\lceil c/(\err-\delta)\rceil$, for some $\delta<\err$,
and $T=\diri=$ the Dirichlet length distribution with $t=2$ and $d=1$. Then, $nM=O\big(c/(\err-\delta)\cdot1/\delta\big)$, using $M=\lceil\sqrt[t-1]{1/\delta}\rceil+d-1$ as suggested  in Lemma~\ref{lem:dirichlet}. 
Then, the minimum value of $\big(c/(\err-\delta)\cdot1/\delta\big)$, as $\delta$ varies, is exactly equal to $\err/2$. 
\end{remark}

\begin{remark}
The above theorem holds if we replace ``locally tractable'' with ``tractable''. In this case, we would change the algorithms by omitting the lines involving $M$ and $D$, and using only one selection: $\selectany(\err/2)$. 
However, the displayed version is more detailed and leaves less work to do when we consider specific problem domains---see Section~\ref{sec:concrete}.
\end{remark}

\pssn
\textbf{Version of the algorithms for instances with finite domains.}
When the problem instances \sspec are descriptions of subsets of a finite domain $X=X_{\sspec}$ 
having polynomially samplable distributions, the algorithms in Fig.~\ref{fig:prax} can be simplified by omitting references related to the truncated distribution---see Fig.~\ref{fig:fin:prax}.
However, we need to assume the existence of a polynomial algorithm 
$\distrparameter(\sspec)$ that returns the index $k$ of the finite distribution $\dt_k$ associated to the instance \sspec.
For example, if \sspec is a block NFA of some word length $\ell$, then the algorithm $\distrparameter(\sspec)$ would return the value $\ell$---which can be computed in time $O(|\sspec|)$.

\begin{figure}[t]
\begin{center}
\parbox{0.4\textwidth}
{
$\finemptiness_{\dt}\,(\sspec,\err)$
\pnsn \hspace*{1\algoindent}
$c := 4.76603$;
\pnsn \hspace*{1\algoindent}
$k := \distrparameter(\sspec)$;
\pnsn \hspace*{1\algoindent}
$n := \lceil c/(\err/2)\rceil$;
\\ \hspace*{1\algoindent}
repeat $n$ times:   
\\ \hspace*{2\algoindent}
$x=\sample(k)$;
\\ \hspace*{2\algoindent}
if $\big(x\in\lang(\sspec)\big)$ 
\\ \hspace*{3\algoindent}
return \false
\\ \hspace*{1\algoindent} 
return \true;
}
\quad
\parbox{0.4\textwidth}
{
$\finuniversality_{\dt}\,(\sspec,\err)$
\pnsn \hspace*{1\algoindent}
$c := 4.76603$;
\pnsn \hspace*{1\algoindent}
$k := \distrparameter(\sspec)$;
\pnsn \hspace*{1\algoindent}
$n := \lceil c/(\err/2)\rceil$;
\\ \hspace*{1\algoindent}
repeat $n$ times:   
\\ \hspace*{2\algoindent}
$x=\sample(k)$;
\\ \hspace*{2\algoindent}
if $\big(x\notin\lang(\sspec)\big)$ 
\\ \hspace*{3\algoindent}
return \false
\\ \hspace*{1\algoindent} 
return \true;
}
%
\caption{
PRAX algorithms for the emptiness problem $E_{\dt}$ (on the left) and for the universality problem $U_{\dt}$ (on the right), in which the instances are descriptions of finite subsets and \dt is a polynomially samplable distribution family. 
}\label{fig:fin:prax}
\end{center}
\end{figure}

\begin{corollary}\label{cor:finite:prax}
Let $\dt=(\dt_k)$ be a polynomially samplable  family of finite distributions.
Let $E_{\dt}$ be any emptiness problem  and let $U_{\dt}$ be any universality problem such that 
(i) the problem instances \sspec are  descriptions of finite subsets with polynomially decidable $\lang(\sspec)$; 
(ii) there is a polynomial algorithm $\distrparameter(\sspec)$ that returns the index $k$ of the distribution $\dt_k$ that matches the domain of $\sspec$: $X_{\sspec}=\dom\dt_k$.
Then, the two algorithms in Fig.~\ref{th:general:prax} are PRAX algorithms  (relative to \dt) for $E_{\dt}$ and $U_{\dt}$.
\end{corollary}

\section{PRAX Algorithms for Some Concrete Problems}\label{sec:concrete}
Here, we consider a few specific emptiness and universality problems in various domains, where we can give concrete complexity estimates  and provide test results of their performance.

\subsection{NFA Universality Revisited and Tautology Testing}\label{sec:nfa:univ}
In \cite{KoMaMoRo:2023}, the authors show a PRAX algorithm for block NFA universality that works in time $O\big(\ell\,|\aut|(1/\err)^2\big)$, where $\ell$ is the word length of $\aut\in\BNFA$. 
The factor $(1/\err)^2$ is simply the number $n$ of words to sample. Here we have the following improvement as a consequence of Theorem~\ref{th:newbound} and Corollary~\ref{cor:finite:prax}.

\begin{corollary}\label{cor:block:nfa}
	There is a PRAX algorithm for block NFA universality that works in time $O\big(\ell\,|\aut|(1/\err)\big)$, where $\ell$ is the word length of the given NFA \aut.
\end{corollary}

\pbsn
\textbf{Tautology testing.}
Testing whether a CNF proposition \sspec with some $k$ variables is a tautology is a universality problem $U_{\ubfam}$: whether all $2^k$ truth assignments satisfy \sspec, or equivalently, whether $\ubfam(\sspec)=1$, where recall that \ubfam is the  family of uniform block distributions. 
The problem of testing whether a proposition is a tautology  is mentioned in \cite{Fortnow:2022} as a good candidate for an interesting problem not in the class \classP. 
The approximate version $U_{\ubfam,\err}$ of $U_{\ubfam}$ is whether $\ubfam(\sspec)\ge1-\err$.
This approximate version could be useful in a non rigid decision making scenario where it is acceptable to know whether a proposition  $\alpha\to\beta$ is true in ``most cases'' (i.e., the ratio of the  satisfying truth assignments over all truth assignments is $\ge1-\err$). 
In $\alpha\to\beta$, which can easily be converted to CNF, $\alpha$ could represent the set of premises and $\beta$ could be the possible action/decision to make.
From the point of view of this paper, the PRAX algorithm for block NFA universality, where the NFAs are over the binary alphabet, can be turned to a PRAX algorithm for testing propositional tautology by simply changing the implementation of the  test ``$x\in\lang(\sspec)$'' from NFA membership to truth assignment testing. Thus, instead of the cost $O(\ell|\aut|)$ of NFA membership, we have the cost $O(|\sspec|)$ of testing whether the CNF proposition \sspec is true.

\begin{corollary}\label{cor:tautology}
	There is a PRAX algorithm for the problem of tautology testing that works in time $O\big(|\sspec|(1/\err)\big)$.
\end{corollary}

\subsection{Emptiness and Universality of 2D Automata}\label{sec:2D}
Here we consider problem instances \sspec which are 2D automata and for which the membership problem is polynomially decidable---in fact, in the class \NL \cite{Lindgren1998Complexity2DPatterns,Smith2019TwoDimensionalAutomata,Smi:2021}.
Each instance \sspec describes a 2D language $\lang(\sspec)$ and implies the 2D length distribution $\diri^2$ with domain $\N_0\times\N_0$ of Example~\ref{ex:dirichlet} as well as the corresponding 2D Dirichlet word distribution $\lbdu{\diri^2}$ with domain $\als^{**}$.

\begin{lemma}\label{lem:2D}
The distribution family $\dt=\lbdu{\diri^2}$ is locally tractable, using the size function 
$\sz:\als^{**}\to\N_0$ with $\sz(z)=\max(|z|_{\mathrm R},|z|_{\mathrm C})$.
Specifically, $\probd{\dt}(m)$ works in time $O(m)$, $\szselect_{\dt}(m)$ works in time $O(m^2)$, and $\maxlen_{\dt}(\delta)$ returns an $M$ in $O\sqrt[t-1]{(1/\delta)}$.
\end{lemma}
\begin{proof}
First, we show the existence of the required algorithm $\probd{\dt}(m)$.
Given a size $m\in\N_0$, the set of 2D words of size $m$ is
\[
\szinv(m)=
\als^{m\times0}\cup \als^{m\times1}\cup\cdots\cup\als^{m\times m-1}
\>\cup\>
\als^{0\times m}\cup \als^{1\times m}\cup\cdots\cup\als^{m-1\times m}
\>\cup\>
\als^{m\times m}.
\]  
Let $A_m= \dt\big(\szinv(m)\big)$ and
let $S_m= \N_0^m\times \N_0^{<m}\cup \N_0^{<m}\times\N_0^m\cup \{(m,m)\}$.
Then,
$$
A_m=\diri(S_m)=\;\diri(m)\diri\big(\N_0^{< m}\big) + \diri(m)\diri\big(\N_0^{< m}\big)+\diri(m)\diri(m),
$$ 
which can be computed in time $O(m)$.  
Hence, the required algorithm $\probd{\dt}(m)$ simply computes the above value $A_m$ in time $O(m)$.
\pnsi
Now we show the existence of the algorithm $\szselect_{\dt}(m)$.
The following distribution $\dd_m$ 
\[
\Big(
\frac{\diri(m)\diri(0)}{A_m},\cdots,
\frac{\diri(m)\diri(m-1)}{A_m},
\frac{\diri(0)\diri(m)}{A_m},\cdots,
\frac{\diri(m-1)\diri(m)}{A_m},
\frac{\diri(m)^2}{A_m}
\Big)
\]
with domain $S_m$  can be computed in time $O(m)$. Thus, the required algorithm $\szselect_{\dt}(m)$ works as follows: (i) compute the above finite distribution $\dd_m$; (ii) invoke $\selectfin(\dd_m)$ to get a pair $(k,\ell)$ with $k,\ell\le m$---this works in time $O(m)$; (iii) $k\ell$ times uniformly select each symbol of the required 2D word---this works in time $O(m^2)$. Hence,  $\szselect_{\dt}(m)$ works in time $O(m^2)$.
\pnsi
Lastly, for  given $\delta\in(0,1)$, we need to find an $M$ such 
that $T(\szinv(\N_0^{>M}))\le\delta$, where recall that 
$\szinv(\N_0^{>M})=\{z\in\als^{**}:\max(|z|_{\mathrm R},|z|_{\mathrm C})\le M\}$.
We have that $T(\szinv(\N_0^{>M}))=\diri^2\big(\N_0^{>M}\times\N_0^{}\big)+\diri^2\big(\N_0^{\le M}\times\N_0^{>M}\big) 
=\diri(\N_0^{>M})+\diri(\N_0^{\le M})\diri(\N_0^{>M})
$; hence,
\[
T\big(\szinv(\N_0^{>M})\big)=x+(1-x)x=2x-x^2,\quad\text{where } x= \diri(\N_0^{>M}).
\]
As $x<1$, we have that $2x-x^2\le \delta$ is equivalent to $x\le 1-\sqrt{1-\delta}$. 
By Lemma~\ref{lem:dirichlet}, if we pick 
\begin{equation}\label{eq:2D:M}
	M\ge\sqrt[t-1]{\frac{1}{1-\sqrt{1-\delta}}}+(d-1)
\end{equation}
then $\diri(\N_0^{>M})=x\le 1-\sqrt{1-\delta}$, as required.
We also note that $M$ is $O\sqrt[t-1]{(1/\delta)}$, as 
$1/\delta<1/(1-\sqrt{1-\delta})<2/\delta$. 
\end{proof}

\begin{corollary}\label{cor:2D:prax}
There is a PRAX algorithm for the emptiness (and for the universality) problem of 2D automata, relative to the 2D word distribution $\lbdu{\diri^2}$, which works in time
\[
O\big(1/\err\cdot\sqrt[t-1]{1/\err^2}\cdot s\big),
\]
where $s$ is the number of states of the 2D automaton used as input to the algorithm.
\end{corollary}
\begin{proof}
Lindgren et al.~\cite{Lindgren1998Complexity2DPatterns} established that one can test membership of a 2D word of dimension $k \times \ell$ by a 2D automaton \sspec having $s$ states in time $O(k\ell s)$.
The statement is a consequence of the PRAX algorithms in Theorem~\ref{th:general:prax} and of Lemma~\ref{lem:2D}, using the facts that (i) $M=O(\sqrt[t-1]{1/\err})$, and (ii) the 2D word $x$ that is tested for membership by \sspec in Theorem~\ref{th:general:prax} is such that $k,\ell\le M$ and, therefore, $k\ell\le M^2$.
\end{proof}


\subsection{Diophantine Equations}\label{sec:diophantine}
Several simple looking Diophantine equations are considered in \cite{Grechuk:2022}, for which it is an open question whether they have any integer solutions. 
Many of these equations have only two or three variables. In some cases the open question reduces to whether the equation and some symmetric one have \emph{positive integer solutions}. 
For example, consider Equation~(66) of \cite{Grechuk:2022}: $x^3y^2-z^3-6=0$.
It is easy to establish the following statement by inspection of cases.

\begin{remark}\label{rem:pos:solution}
	Equation $x^3y^2-z^3-6=0$ has some integer solution if and only if at least one of  $x^3y^2-z^3-6=0$ and $x^3y^2-z^3+6=0$ has a positive integer solution.
\end{remark}
\pnsi
Here we consider a PRAX algorithm for the problem of whether a 3-variable Diophantine equation $\sspec=\sspec(x,y,x)$ has nonnegative integer solutions.
The distribution implied by \sspec is $\diri^3$ such that
\[
\diri^3(j,k,\ell) = \diri(j)\cdot \diri(k)\cdot \diri(\ell).
\]
We assume that the algorithmic size $|\sspec|$ of the equation \sspec is the number of terms in the equation, and that evaluating whether a triple $(n_1,n_2,n_2)$ satisfies the equation requires linear time $O(|\sspec|)$.

\begin{lemma}\label{lem:dioph}
The distribution family $\dt=\diri^3$ is locally tractable, using the size function 
$\sz:\N_0\times\N_0\times\N_0\to\N_0$ with $\sz(j,k,\ell)=\max(j,k,\ell)$.
Specifically, $\probd{\dt}(m)$ works in time $O(m)$, $\szselect_{\dt}(m)$ works in time $O(m^2)$, and $\maxlen_{\dt}(\delta)$ returns an $M$ in $O\sqrt[t-1]{(1/\delta)}$.
\end{lemma}
\begin{proof}
First, we show the existence of the required algorithm $\probd{\dt}(m)$.
Given a size $m\in\N_0$, the set of triples $(j,k,\ell)$ of size $m$ is
\begin{align*}
\szinv(m) = &\{m\}\times\N_0^{<m}\times\N_0^{<m} \bigcup
            \N_0^{<m}\times\{m\}\times\N_0^{<m} \bigcup
            \N_0^{<m}\times\N_0^{<m}\times\{m\} \bigcup\\
          &\{m\}\times\{m\}\times\N_0^{<m} \bigcup
            \{m\}\times\N_0^{<m}\times\{m\} \bigcup
            \N_0^{<m}\times\{m\}\times\{m\} \bigcup \\
          &  \{(m,m,m)\}.
\end{align*}
Let $y_m=\diri(\N_0^{<m})$, which can be computed in time $O(m)$.
Then, $\dt\big(\szinv(m)\big)=3y_m^2+3y_m+\diri(m)^3$, which is the output of $\probd{\dt}(m)$. 
\pnsi
Now we show the existence of the algorithm $\szselect_{\dt}(m)$.
The required algorithm needs to select a triple from $\szinv(m)$, which is shown above, according to the finite distribution $\big(T|\szinv(m)\big)$.
The domain $\szinv(m)$ of $\big(\dt|\szinv(m)\big)$ has $3m^2+3m+1$ triples  
and the probability of each triple $(j,k,\ell)$ is 
$\big(\diri(j)\cdot \diri(k)\cdot \diri(\ell)\big)/\dt\big(\szinv(m)\big)$.
Thus, the required algorithm $\szselect_{\dt}(m)$ works as follows: (i) compute the  finite distribution $\dd_m=\big(T|\szinv(m)\big)$; (ii) invoke $\selectfin(\dd_m)$ to get a triple $(j,k,\ell)$ with $j,k,\ell\le m$---this works in time $O(m^2)$. Hence,  $\szselect_{\dt}(m)$ works in time $O(m^2)$.
\pnsi
Lastly, for  given $\delta\in(0,1)$, we need to find an $M$ such that
$\dt(\szinv(\N_0^{>M}))\le\delta$, where recall that 
$\szinv(\N_0^{>M})=\{(j,k,\ell): \max(j,k,\ell)>M\}$.
We have that
\[
\szinv(\N_0^{>M})=
\N_0^{>M}\times\N_0\times\N_0\bigcup
\N_0^{\le M}\times\N_0^{>M}\times\N_0\bigcup
\N_0^{\le M}\times\N_0^{\le M}\times\N_0^{>M}.
\]
Let $x=\diri(\N_0^{>M})$.
Then,  $\dt(\szinv(\N_0^{>M}))=x+(1-x)x+(1-x)^2x$. 
Let $f(x)=x^3-3x^2+3x-\delta$.
Then, $T(\szinv(\N_0^{>M}))-\delta=f(x)$ and we want $M$ such that $f(x)\le0$ with $x\in(0,1)$.
However we first study $f(x)$ across $[-\infty,+\infty]$.
We have $f'(x)=3(x-1)^2$, which is always $\ge0$ with $x=1$ as the only zero; $f''(x)=6(x-1)$ which is zero at $x=1$; and $f'''(x)=6>0$. Hence, $x=1$ is an inflection point of $f(x)$, and as there are no other zeros of $f'(x)$, $f(x)$ has no local minima or maxima. 
Moreover, as $f(x)$ goes to $-\infty$ when $x\to-\infty$ and to $+\infty$ as $x\to+\infty$, we have that (i) $f(x)$ is decreasing in $(-\infty,1)$; (ii) it is increasing in $(1,\infty)$;  (iii) it has exactly one root $r$, with $r\in(0,1)$ as $f(1)>0$ and $f(0)<0$; (iv) $f(x)<0$ for $x<r$ and $f(x)>0$ for $x>r$.
Moreover, one can verify that $f(\delta/2)<0$ and $f(3\delta/2)>0$; hence,
\[
\delta/2<r<3\delta/2 \quad\leftrightarrow\quad 2/\delta>1/r>2\delta/3.
\]
By Lemma~\ref{lem:dirichlet}, 
if we pick $M\ge \sqrt[t-1]{2/\delta}+(d-1)$ then 
$x=\diri(\N_0^{>M})<r$, and then $f(x)<0$, which
implies $T(\szinv(\N_0^{>M}))<\delta$, as required.
\end{proof}

\begin{corollary}\label{cor:3var:dioph}
There is a PRAX algorithm for the emptiness (and for the universality) problem of Diophantine equations \sspec with three nonnegative variables, relative to the 3D length distribution ${\diri^3}$, which works in time
\[
O\big(\sqrt[t-1]{1/\err^2} +\, 1/\err\cdot\sqrt[t-1]{1/\err} +1/\err\cdot|\sspec|\big).
\]
\end{corollary}
\begin{proof}
Using the PRAX algorithms in Theorem~\ref{th:general:prax} and  Lemma~\ref{lem:dioph}, we have the time complexity
$O\big(1/\err\cdot\sqrt[t-1]{1/\err^2} +1/\err\cdot|\sspec|\big)$.
However, the body of the loop ``repeat $n$ times'' in Theorem~\ref{th:general:prax}, can be replaced with the simpler 
\pssi\quad $x := \selectfin(D)$
\pnsi\quad if $x=\none$ continue 
\pnsi\quad $y := \selectfin(D)$
\pnsi\quad if $y=\none$ continue
\pnsi\quad $z := \selectfin(D)$
\pnsi\quad if $z=\none$ continue
\pnsi\quad if $(x,y,z)$ satisfies \sspec return \false
\pssn
It is easy to see that the independent selection  of three integers in $[0,M]$, or \none, is equivalent to the selection specified in the proof of Lemma~\ref{lem:dioph}. The time complexity of the algorithm becomes as required.  
\end{proof}

%
\pnsn
\textbf{Testing emptiness of $\bm{x^3y^2-z^3-6=0}$.}
In view of Remark~\ref{rem:pos:solution}, we can use the PRAX emptiness algorithm on the equations $x^3y^2-z^3-6=0$ and $x^3y^2-z^3+6=0$ relative to the distribution ${\diri^3}$ for $d=1$. In fact, in view of the following lemma, we use $d=2$.

\begin{lemma}\label{lem:d:at:least2}
	If either of $x^3y^2-z^3-6=0$ and $x^3y^2-z^3+6=0$ has a positive integer solution $(x,y,z)$ then  $x,y,z\ge 2$.
\end{lemma}
\begin{proof}
    Suppose that $(x,y,z)$ is a positive integer solution: $x^3y^2=z^3+6$ or $x^3y^2=z^3-6$.
	If $z=1$, then $x^3y^2=7$ or $x^3y^2=-5$, which is not possible for any $x,y\ge1$. Hence, $z\ge2$. If $x=1$, then both $y^2=z^3+6$ and $y^2=z^3-6$ are Mordell Diophantine equations: equations of the form $y^2=z^3+ k$, with $k$ being an integer~\cite{BennGhad:2015}. 
	But neither of these two equations has a solution for $k=\pm6$; see sequences A054504 and A081121 in the OEIS~\cite{OEIS}. Hence, $x\ge2$. Finally, if $y=1$, then $|x^3-z^3|=6$. But one can verify that the difference  $|x^3-z^3|$ is always greater than $6$. Hence, $y\ge2$.
\end{proof}

\begin{table}[t]
\caption{
Testing \err-emptiness of solutions to $x^3y^2-z^3-6=0$ and $x^3y^2-z^3+6=0$ relative to three Dirichlet distributions $\diri^3$: (i) with parameters $\err=0.00001,t=2.0,d=2$, hence $M=400,001$; (ii) with parameters $\err=0.001,t=1.5,d=2$, hence $M=16,000,001$; (iii) with parameters $\err=0.05,t=1.25,d=2$, hence $M=40,960,001$. 
For each case, the algorithm was run five times and returned \true in all cases (no solutions  found). 
As $t$ gets closer to 1, the maximum integer $M$ that can be sampled gets larger.
Machine used: MacBook Pro, M2 Max, 64 GB memory. Programming: Python 3, using PRAX implementation in \cite{Fado}.
}
\label{table:test:dioph}
\centering
  \begin{footnotesize}
\begin{tabular}{rc} \hline
$t=2.0,\err=10^{-5}$  & Answer
   \\\hline
107.159 sec 	&  \true
\\
108.263 sec 	&  \true
\\
104.091 sec 	&  \true
\\
115.994 sec 	&  \true
\\
108.554 sec 	&  \true
\\\hline
\end{tabular}
\quad
\begin{tabular}{rc} \hline
$t=1.5,\err=0.001$  & Answer
   \\\hline
10,260.828 sec 	&  \true
\\
9,060.668 sec 	&  \true
\\
8,396.982 sec 	&  \true
\\
10,481.317 sec 	&  \true
\\
7,289.750 sec 	&  \true
\\\hline
\end{tabular}
\quad
\begin{tabular}{rc} \hline
$t=1.25,\err=0.05$  & Answer
   \\\hline
49,111.178 sec 	&  \true
\\
28,295.932 sec 	&  \true
\\
24,449.394 sec 	&  \true
\\
17,735.975 sec 	&  \true
\\
20,290.405 sec 	&  \true
\\\hline
\end{tabular}
\end{footnotesize}
\end{table}

\begin{remark}\label{rem:equat:test}\textsf{[Experimental ``Theorem'']}
	The three PRAX experiments in Table~\ref{table:test:dioph} have found no solutions to the equations $x^3y^2-z^3\pm6=0$. By Corollary~\ref{cor:3var:dioph}, this implies that: 
	\begin{enumerate}
		\item with probability $\ge1023/1024$, the set $S$ of solutions (of either equation) is $0.00001$-close to being empty, relative to ${\diri^3}$ with $t=2.00$ and $d=2$, i.e. ${\diri^3}(S)\le 0.00001$.
		\item with probability $\ge1023/1024$, the set $S$ of solutions (of either equation) is $0.001$-close to being empty, relative to ${\diri^3}$ with $t=1.50$ and $d=2$, i.e. ${\diri^3}(S)\le 0.001$.
		\item with probability $\ge1023/1024$, the set $S$ of solutions (of either equation) is $0.05$-close to being empty, relative to ${\diri^3}$ with $t=1.25$ and $d=2$, i.e. ${\diri^3}(S)\le 0.05$.
	\end{enumerate}
\end{remark}

\section{Concluding Remarks}\label{sec:last}
In this paper, we have extended the PRAX algorithms for NFA universality and NFA equivalence in \cite{KoMaMoRo:2023,KoMoReSe:2024} to work for any universality and emptiness problems whose instances have tractable or samplable domains. 
Although widely applicable, these algorithms are only meaningful where approximation is meaningful. As stated in \cite[pg.~417]{Gold:2008}:
\begin{quotation}
The answer to [what constitutes a ``good'' approximation] seems intimately related to the specific computational task at hand \dots the importance of certain approximation problems is much more subjective \dots [which] seems to stand in the way of attempts at providing a \emph{comprehensive} theory of \emph{natural} approximation problems.
\end{quotation}
In \cite{KoMaMoRo:2023}, NFA universality applies to code maximality where it is acceptable to have a code that is $\err$-close maximal as opposed to strictly maximal.

A PRAX algorithm $A(\sspec,\err)$ fails when $\lang(\sspec)$ is not close to being universal, or empty, relative to some distribution \dt, but the algorithm returns \true.
One of the observations in \cite{KoMaMoRo:2023} was that no NFA universality examples were found for which the algorithm fails; that is, the algorithm would always find a witness of non-universality when the given NFA is not \err-close to being universal.
This led to improving in this paper the sample size bound from quadratic to linear w.r.t. $1/\err$. However, we still have not been able to find examples of instances where the algorithm fails. It might be difficult to further improve the linear sample size bound asymptotically, but perhaps there is a way to reduce the constant of that bound.

In the spirit of experimental mathematics, a large experiment could be set up, possibly involving parallel sampling, to obtain stronger empirical results of the sort in Remark~\ref{rem:equat:test} for values of $t$ closer to 1.


\bibliographystyle{plain}
\bibliography{refs.bib}

\begin{thebibliography}{10}

\bibitem{ArBa:2009}
Sanjeev Arora and Boaz Barak.
\newblock {\em Computational Complexity: A Modern Approach}.
\newblock Cambridge University Press, Cambridge, 2009.

\bibitem{BCGL:1992}
Shai Ben{-}David, Benny Chor, Oded Goldreich, and Michael Luby.
\newblock On the theory of average case complexity.
\newblock {\em Journal of Computer and System Sciences}, 44(2):193--219, 1992.

\bibitem{BennGhad:2015}
Michael~A. Bennett and Amir Ghadermarzi.
\newblock Mordell's equation: a classical approach.
\newblock {\em LMS Journal of Computation and Mathematics}, 18(1):633--646,
  2015.

\bibitem{BorwDevl:2008}
Jonathan Borwein and Keith Devlin.
\newblock {\em The Computer as Crucible: An Introduction to Experimental
  Mathematics}.
\newblock A K Peters/CRC Press, Wellesley, 2008.

\bibitem{CalDum:2020}
Cristian~S. Calude and Monica Dumitrescu.
\newblock A statistical anytime algorithm for the halting problem.
\newblock {\em Computability}, 9(2):155--166, 2020.

\bibitem{Fortnow:2022}
Lance Fortnow.
\newblock Fifty years of {P} vs. {NP} and the possibility of the impossible.
\newblock {\em Communications of the {ACM}}, 65(1):76--85, 2022.

\bibitem{GiammarresiRestivo19972DLanguages}
Dora Giammarresi and Antonio Restivo.
\newblock Two-dimensional languages.
\newblock In Grzegorz Rozenberg and Arto Salomaa, editors, {\em Handbook of
  Formal Languages}, volume~3, pages 215--267. Springer-Verlag, Berlin, 1997.

\bibitem{Gold:2008}
Oded Goldreich.
\newblock {\em Computational Complexity: A Conceptual Perspective}.
\newblock Cambridge University Press, Cambridge, 2008.

\bibitem{Gol:1970}
Solomon~W. Golomb.
\newblock A class of probability distributions on the integers.
\newblock {\em Journal of Number Theory}, 2(2):189--192, 1970.

\bibitem{Gol:1992}
Solomon~W. Golomb.
\newblock Probability, information theory, and prime number theory.
\newblock {\em Discrete Mathematics}, 106--107:219--229, 1992.

\bibitem{Grechuk:2022}
Bogdan Grechuk.
\newblock On the smallest open diophantine equations.
\newblock {\em {SIGACT} News}, 53(1):36--57, 2022.

\bibitem{HeZhaZha:2010}
Simai He, Jiawei Zhang, and Shuzhong Zhang.
\newblock Bounding probability of small deviation: {A} fourth moment approach.
\newblock {\em Mathematics of Operations Research}, 35(1):208--232, 2010.

\bibitem{HoMoUl:2001}
John~E. Hopcroft, Rajeev Motwani, and Jeffrey~D. Ullman.
\newblock {\em Introduction to Automata Theory, Languages, and Computation}.
\newblock Addison-Wesley, Boston, 2nd edition, 2001.

\bibitem{OEIS}
{OEIS}~Foundation Inc.
\newblock The {O}n-{L}ine {E}ncyclopedia of {I}nteger {S}equences.
\newblock Published electronically at \url{http://oeis.org/}.

\bibitem{Inoue19912DAutomataSurvey}
Katsushi Inoue and Itsuo Takanami.
\newblock A survey of two-dimensional automata theory.
\newblock {\em Information Sciences}, 55(1--3):99--121, 1991.

\bibitem{KoMaMoRo:2023}
Stavros Konstantinidis, Mitja Mastnak, Nelma Moreira, and Rog{\'{e}}rio Reis.
\newblock Approximate {NFA} universality and related problems motivated by
  information theory.
\newblock {\em Theoretical Computer Science}, 972:114076, 2023.

\bibitem{KMR:2018}
Stavros Konstantinidis, Nelma Moreira, and Rog{\'e}rio Reis.
\newblock Randomized generation of error control codes with automata and
  transducers.
\newblock {\em RAIRO - Theoretical Informatics and Applications},
  52(2--3--4):169--184, 2018.

\bibitem{KoMoReSe:2024}
Stavros Konstantinidis, Nelma Moreira, Rog\'erio Reis, and Juraj {\v S}ebej.
\newblock On the difference set of two transductions.
\newblock Technical report 2024-01, Mathematics and Computing Science, Saint
  Mary's University, Halifax, Canada, 2024.

\bibitem{Lindgren1998Complexity2DPatterns}
Kristian Lindgren, Cristopher Moore, and Mats Nordahl.
\newblock Complexity of two-dimensional patterns.
\newblock {\em Journal of Statistical Physics}, 91(5/6):909--951, 1998.

\bibitem{MiUp:2017}
Michael Mitzenmacher and Eli Upfal.
\newblock {\em Probability and Computing: Randomization and Probabilistic
  Techniques in Algorithms and Data Analysis}.
\newblock Cambridge University Press, Cambridge, 2nd edition, 2017.

\bibitem{Fado}
Rog{\'e}rio Reis and Nelma Moreira.
\newblock {FAdo}: {T}ools for {F}ormal {L}anguages manipulation.
\newblock Published electronically at \url{http://fado.dcc.fc.up.pt/}.

\bibitem{FLhandbookI}
Grzegorz Rozenberg and Arto Salomaa, editors.
\newblock {\em Handbook of Formal Languages}, volume~1.
\newblock Springer-Verlag, Berlin, 1997.

\bibitem{Smith2019TwoDimensionalAutomata}
Taylor~J. Smith.
\newblock Two-dimensional automata.
\newblock Technical report 2019-637, School of Computing, Queen's University,
  Kingston, Canada, 2019.

\bibitem{Smi:2021}
Taylor~J. Smith.
\newblock {\em Closure, Decidability, and Complexity Results for Restricted
  Variants of Two-Dimensional Automata}.
\newblock PhD thesis, Queen's University, 2021.

\end{thebibliography}

\end{document}